\documentclass[10pt]{amsart}

\usepackage[mathscr]{eucal}
\usepackage{amsmath,amssymb,amscd,amsthm}%,latexsym
\usepackage{color} 
\usepackage{epsfig}
\newtheorem{Theorem}{Theorem}

\newtheorem{Definition}{Definition}  
\newtheorem{Corollary}{Corollary}

\theoremstyle{definition} 
\newtheorem{Remark}{Remark} 

\theoremstyle{plane}

\def \beq{ \begin{equation} }
\def \eeq{\end{equation}}

\title{Conjugated equilibrium solutions for the $2$--body problem in the two dimensional sphere $\mathbb{M}^2_R$ for equal masses}
   
\begin{document}

\maketitle

\markboth{Ortega-Palencia Pedro Pablo and  Reyes-Victoria J. Guadalupe}{Conjugated equilibria solutions for the $2$--body problem in the two dimensional sphere}

\vspace{-0.5cm}

\author{
\begin{center}
{\rm PEDRO PABLO ORTEGA PALENCIA \\
        Universidad de Cartagena \\
 Departamento de Matem\'aticas  \\
   Campus San Pablo,  CP 130015\\
         Cartagena de Indias, Colombia\\
         {\tt portegap@unicartagena.edu.co}\\
        \bigskip
         J. GUADALUPE REYES VICTORIA \\
            Universidad Aut\'onoma Metropolitana \\
 Departamento de Matem\'aticas \\
        Unidad Iztapalapa, CP 09340,
 Cd. de M\'exico,   M\'exico \\
         {\tt revg@xanum.uam.mx}}
\end{center}}

 \bigskip
\begin{center}
\today
\end{center}

\begin{abstract}
We study here the behaviour of solutions for conjugated (antipodal) points in the $2$-body problem on the two-dimensional  sphere $\mathbb{M}^2_R$. We use a slight modification of the classical potential used commonly  in \cite{Borisov}, \cite{Diacu} and \cite{Perez}, which avoids the conjugated (antipodal) points as singularities and permit us obtain solutions through these points, as limit of relative equilibria. Such limit solutions behave as relative equilibria because are invariant under Killing vector fields in the Lie Algebra ${\rm su} (2)$ and are geodesic curves.
\end{abstract}

\medskip

\footnote*{MSC: Primary 70F15, Secondary 53Z05}

\smallskip

\footnote*{Keywords: Two dimensional conformal sphere $\mathbb{M}^2_R$,  The $n$--body problem, Relative equilibria, Conjugated points.}

\section{Introduction}
\label{sec:intro}

We consider in this paper the problem of studying the motion of two point interacting particles of masses $m_1,\cdots,m_n$  on the 
two-dimensional conformal sphere $\mathbb{M}^2_R$ under the action of a suitable (cotangent) potential.  

 We omit along all the document the term {\it curved} because the studied space has a non-euclidian metric. Also, we omit the term 
{\it intrinsic}, since when  it is chosen the geometric structure on the conformal sphere $\mathbb{M}_R^2=\widehat{\mathbb{C}} = \mathbb{C} \cup  \{ \infty \}$ as  the Riemann surface with the canonical complex variables $(w,\bar{w})$ endowed with the conformal metric
\begin{equation}\label{met-c}
   ds^2= \frac{4R^4 \,dw d\bar{w}}{(R^2 + |w|^2)^2},
\end{equation}
 the corresponding  differentiable structure in such coordinates is given to this space  (see \cite{Dub, Farkas} for more details). 

 Following the methods of the geometric Erlangen program
as in \cite{Kisil} for the space $\mathbb{M}_R^2$, are defined the  conic motions of the $n$--body problem in terms of the
action of one dimensional subgroups of M\"{o}bius isometric transformations group ${\rm SU} (2)$, associated to a suitable Killing  vector field. By the method of matching
the vector field in the Lie algebra associated to the corresponding
subgroup with the cotangent gravitational field we state in each case the functional algebraic
conditions (depending on the time $t$) that the solutions must hold in order to be a relative equilibrium. See  \cite{Diacu2} and
\cite{Perez} where these techniques were also used.

\smallskip

In celestial mechanics on surfaces with positive Gaussian curvature, it is very common for the cotangent potential to be used as an extension of the Newtonian potential. This paper presents geometric, dynamic and analytical arguments that justify the introduction of a subtle variant of this potential. With the help of stereographic projection, it is shown that the potential that best approximates the Newtonian potential for the case of constant and positive Gaussian curvature is $\displaystyle \cot \left(\frac{\theta}{2} \right)$ instead of the one normally used, $\cot \theta $.

\smallskip

This paper is  organized as follows.

\smallskip

In Section  \ref{sec:equamotion} we obtain the new potential for the problem which avoids the conjugated (antipodal) points as singularities and are stated  the equations of motion of the problem in  complex coordinates in $\mathbb{M}^2_R$ as in \cite{Diacu2} and \cite{Perez}. 

\smallskip

In Section \ref{elliptic} we show the algebraic equations for the elliptic relative equilibria for the general problem with the new potential.

\smallskip

In section \ref{sec:two-body} we obtain the classification of the relative equilibria for the two body problem as {\it Borisov et al} in \cite{Borisov}.

In section \ref{sec:conjugated points} we obtain the limit solutions for any conjugated (antipodal) pair of points  in $\mathbb{M}^2_R$, which satisfy one regularized system of  equations of motion obtained from the original one.

%%%%%%%%%%%%%%%%%%%%%%%%%%%%%%%%%%%%%%%%%%%%%%%%%%%%%
\section{Equations of motion and relative equilibria}\label{sec:equamotion}
%%%%%%%%%%%%%%%%%%%%%%%%%%%%%%%%%%%%%%%%%%%%%%%%%%%%%%%%%

In \cite{Perez} the authors obtain the equations of motion for this problem using the stereographic projection of the sphere (of radius $R$) embedded in $\mathbb{R}^3$ into the complex plane $\mathbb{C}$ endowed with the  metric (\ref{met-c}). Moreover, in \cite{Abraham-Marsden}  the classical equations of  the 
of particles with positives masses $m_1,m_2,\cdots,m_n$  in a Riemannian or semi-Remannian manifold with coordinates $(x^1, x^2, \cdots, x^N)$ endowed with a metric $(g_{ij})$ and associated connection $\Gamma^i_{jk}$, and such that the particles move under the influence  of a pairwise acting potential $U$ are given by
\begin{equation}\label{eq:Vlasov-Poisson}
 \frac{D \dot{x}^i}{d t} =\ddot{x}^i+ \sum_{l,j} \Gamma^i_{lj} \dot{x}^l \dot{x}^j = \sum_{k} m_k g^{ik} \frac{\partial U}{\partial x^k},
\end{equation} 
for $i=1,2, \cdots, N$, where $\displaystyle \frac{D}{d t}$ denotes the covariant derivative and $g^{-1}=(g^{ik})$ is the inverse matrix for the metric $g$.
\begin{Remark}
We observe that in equation (\ref{eq:Vlasov-Poisson}), the left hand
side corresponds to the equation of the geodesic curves, whereas the right hand side corresponds to the gradient of the potential in the given metric. This means  that if the potential is constant, then the particles move along geodesics.
\end{Remark}

We state in this section the equations of motion for the  $n$-body problem in the conformal sphere $\mathbb{M}^2_R$.

\subsection{The new potential}

In the study of the $n$-body problem on the sphere the potential has been considered
\begin{equation}\label{badpotential}
U_{0}(\theta)=\cot(\theta),
\end{equation}
where $\theta$ represents the angle at the center of the sphere by the position vectors of the particles. This potential is attractive for values of $\theta \in (0,\pi/2)$, and it is repulsive if $\theta \in (\pi/2,\pi)$,  and vanishes in $\pi/2$, which does not correspond to the properties of the Newtonian potential, which is always attractive.

Consider, two particles with masses $m_{k}, m_{j}$ sited at the points  $A_{k}=(x_{k},y_{k},z_{k})$, $A_{j}=(x_{j},y_{j},z_{j})$ on the sphere $\mathbb{S}^2_R$, $d_{kj}$ their geodesic distance  and $\theta_{kj} \in [0,\pi]$ the angle at the origin between these. Then,
$$\cos\theta_{kj}=\frac{A_{k} \cdot A_{j}}{R^{2}}$$
and $d_{kj}=\theta_{kj}R$. 

From basic trigonometric relationships in the sphere $\mathbb{S}^2_R$ we get,
\begin{equation}\label{cotangente1}
\cot \left(\frac{d_{kj}}{2R}\right)=\sqrt{\frac{R^{2}+A_{j}\cdot A_{k}}{R^{2}-A_{j}\cdot A_{k}}}
\end{equation}
where
$A_{k}\cdot A_{j}=x_{k}x_{j}+y_{k}y_{j}+z_{k}z_{j}$.

After of applying the stereographic projection on $\mathbb{M}^2_R$, we obtain
$$A_{k}\cdot A_{j}=R^{2} \left(\frac{2R^{2}(w_{k}\bar{w}_{j}+\bar{w}_{k}w_{j})+(R^{2}-|w_{k}|^{2})(R^{2}-|w_{j}|^{2})}{(R^{2}+|w_{k}|^{2})(R^{2}+|w_{j}|^{2})}\right)$$ 
If we put
\begin{eqnarray}
a &=&2R^{2}(w_{k}\bar{w}_{j}+\bar{w}_{k}w_{j})+(R^{2}-|w_{k}|^{2})(R^{2}-|w_{j}|^{2}) \nonumber \\
b &=&(R^{2}+|w_{k}|^{2})(R^{2}+|w_{j}|^{2}), \nonumber \\
\end{eqnarray}
then the equation (\ref{cotangente1}) becomes,
$$\cot \left(\frac{d_{kj}}{2R}\right)=\sqrt{\frac{b+a}{b-a}}
= \frac{1}{R}\frac{|R^{2}+w_{k}\bar{w}_{j}|}{|w_{k}-w_{j}|}, $$
since $b+a=2|R^{2}+w_{k}\bar{w}_{j}|^{2}$ y $b-a=2R^{2}|w_{k}-w_{j}|^{2}$.

Therefore, if two particles with positive masses  $m_{k},m_{j}$ are located at the points $A_{k},A_{j}$ on the sphere $\mathbb{S}^2_R$, with respective stereographic projections $w_{k},w_{j}$, they will be under the action of the potential
\begin{equation}\label{goodpotential2}
U^{kj}_{R}=\frac{m_{j}m_{k}}{R}\cot\left(\frac{d_{kj}}{2R}\right)=\frac{m_{j}m_{k}|R^{2}+w_{k}\bar{w}_{j}|}{R^{2}|w_{k}-w_{j}|}
\end{equation}

\begin{Corollary}
When $R\rightarrow\infty$, the pairwise acting potential $U^{kj}_{R}$ converges to the Newtonian potential on the complex plane.
\end{Corollary} 
 
The authors at \cite{Perez} have used the pairwise acting potential (see \cite{Perez}),
$$V^{kj}_{R}=m_{k}m_{j}\cot \left(\frac{d_{kj}}{R}\right)=m_{k}m_{j}\cot (\theta_{kj})$$

Note that when $ R \to \infty $ the pairwise acting potential $V^{kj}_{R}$ also converges to the classic Newtonian potential. However, such potential generates singularities due to collisions and antipodal configurations, while 
(\ref{goodpotential2}) only presents singularities in collisions. In addition, the equations of motion that will be obtained using the potential (\ref{goodpotential2}) are clearer and more precise than those obtained with the $V^{kj}_{R}$ potential, as can be seen in \cite{Perez}.

\subsection{Equations of motion}

Let us denote by $\mathbf{w}=(w_1, w_2,\cdots,w_n) \in  (\mathbb{M}^2_{R})^n$  the total vector position of $n$ particles with masses
$m_i>0$ located at point $w_i, i=1,2, \cdots, n$,  on the space $\mathbb{M}^2_{R}$.

The singular set in  $\mathbb{M}^2_{R}$ for the $n$-body problem given by the cotangent relation (\ref{goodpotential2}) is the set of zeros  of the equation
 \[ w_j-w_k =0. \]

From here, are obtained the following singular  set 
$$\Delta(C) = \cup_{kj} \,\,  \Delta(C)_{kj},$$
where
\[ \Delta(C)_{kj}=\{\mathbf{w} = (w_1, w_2,\cdots, w_n) \in (\mathbb{M}^2_{R})^n \, | \, w_k = w_j, \, k \neq j\} \]
corresponds to the pairwise collision of the particles with masses $m_j$ and $m_k$.

It is clear that the presence of geodesic conjugated points on an arbitrary Riemannian  manifold  does not allows us in general to singularities in the  equations of motion (\ref{eq:Vlasov-Poisson}) of such mechanical system, by the same reason as in this particular case.

\smallskip

From now on, we will suppose that the  $n$ point particles in the space $\mathbb{M}^2_{R}$, are moving under the action of the potential 
\begin{equation}\label{eq:potesf}
U_R (\mathbf{w},\mathbf{\bar{w}}) = \frac{1}{R^{2}} \sum_{j < k}^n m_k m_j \frac{|R^2+ w_j \bar{w}_k|}{|w_j-w_k|}, 
\end{equation}
defined in the set $(\mathbb{M}^{2}_{R})^n \setminus \Delta$.

\smallskip

A direct substitution of the equations of the geodesics curves in $\mathbb{M}^2_{R}$ and the gradients
\begin{eqnarray}\label{eq:gradmet}
\frac{\partial U_R}{\partial \bar{w}_k} \,(\mathbf{w},\mathbf{\bar{w}}) &=&
 \frac{1}{2R^{2}}\sum_{j < k}^n\frac{m_k m_j(w_j-w_k)(R^2+\bar{w}_jw_k)(R^2+|w_j|^2)}{\, |w_j-w_k|^3 \, |R^2+ \bar{w}_j w_k|} \nonumber \\
 \end{eqnarray}
 in system (\ref{eq:Vlasov-Poisson}),  shows that the  solutions of the problem must  satisfy the following system of second order ordinary differential equations
\begin{equation}\label{eq:motiongral}
 \ddot{w}_{k} - \frac{2\bar{w}_{k}\dot{w}_{k}^{2}}{R^{2}+|w_{k}|^{2}}=\frac{(R^{2}+|w_{k}|^{2})^{2}}{4R^{6}}\sum_{j < k}^{n} m_{j}\frac{(w_{j}-w_{k})(R^{2}+\bar{w_{j}}w_{k})(R^{2}+|w_{j}|^{2})}{|w_{j}-w_{k}|^{3}|R^{2}+\bar{w_{j}}w_{k}|}
\end{equation}
$k=1,2,\dots ,n$, and where,
\begin{equation}\label{eq:conforesf}
\lambda (w_k, \bar{w}_k)= \frac{4R^4}{(R^2+|w_k|^2)^2}
\end{equation}
is the value of conformal function for the Riemannian metric (\ref{met-c})  in the point $(w_k, \bar{w}_k)$ (see \cite{Perez}).

 It has already been observed that when $ R \to \infty $, the  equations of motion are obtained for the $n$-body problem in the plane with Newtonian potential. On the other hand, the set of singularities in $\mathbb{M}^2_{R}$ for the $n$-body problem it is only the set of points of equation $w_j-w_k=0$ (binary collisions of the particles sited  in $w_{j}$ y $w_{k}$ ) while the conjugate points (antipodal points) defined by the equations $R^2+ \bar{w}_j w_k =0$ are not longer singularities. However both sets are singularities for the acting force on the problem.
 
  As it is expected, the behaviour of these two types of points, singularities and zeros of the potential, is completely different, which is not possible to appreciate in the equations of motion obtained in \cite{Perez, Diacu2}, due precisely to the choice of the potential.

%%%%%%%%%%%%%%%%%%%%%%%%%%%%%%%%%%%%%%%%%%%%%%%%%%%%%%%%%
\section{Elliptic relative equilibria}\label{elliptic}
%%%%%%%%%%%%%%%%%%%%%%%%%%%%%%%%%%%%%%%%%%%%%%%%%%%%%%%%%%

We start our analysis of the relative equilibria solutions with the so called, from now on by short, {\it elliptic solutions}, obtained by the action of the canonical  one-dimensional parametric subgroup of $SU(2)$ asociated to the diferential equation $\displaystyle \dot{w}_k= 2 i \, w_k $. The action of the subgroup has  been  partially studied in \cite{Perez}.

We have  the following first result.

\begin{Theorem} \label{theo:relative-equilibria}
Let be $n$ point particles with masses $m_1,m_2, \cdots, m_n>0$
moving in $\mathbb{M}_R^2$. An equivalent condition for $ \mathbf{w}(t)=(w_1(t), w_2(t), \cdots, w_n(t))$  to be an elliptic solution
of (\ref{eq:motiongral}) is that the
coordinates satisfy the following rational functional equations depending on the time.
\begin{equation} \label{eq:rationalsystem-2}
\frac{16 \, R^6(|w_k|^2-R^2) w_k }{ (R^2+ |w_k|^2)^4  } =
 \sum_{j=1, j \neq k}^n \frac{m_j \, (|w_j|^2+R^2)^2(R^2+\bar{w}_jw_k)(w_j-w_k)}{|R^2+ \bar{w}_j w_k| \, |w_j-w_k|^3 }
\end{equation}
with velocity $\displaystyle \dot{w}_k= 2 i \, w_k $ at each point.
\end{Theorem}

\begin{proof} By straightforward computations,  for the first case, we have from equation $\displaystyle \dot{w}_k= 2 i \, w_k $ the equality
\begin{equation}
\ddot{w}_k= -4 \, w_k,
\end{equation}
which  substituted, into equation (\ref{eq:motiongral}),
gives us the relation (\ref{eq:rationalsystem-2}).
\end{proof}

The following result give us conditions on the initial positions of the particles to generate an elliptic solution of equation (\ref{eq:rationalsystem-2}). In other words, such solutions do depend on such fixed points. Let $ \mathbf{w}(0) = (w_1(0), w_2(0), \cdots, w_n(0))=(w_{1,0}, w_{2,0}, \cdots, w_{n,0})$ be the vector of initial positions of the particles.

\begin{Corollary}\label{coro:cyclical-2} With the hypotesis of Theorem \ref{theo:relative-equilibria}, a  necessary and sufficient condition for  the initial positions $w_{1,0}, w_{2,0}, \cdots, w_{n,0}$ to generate an elliptic  solution for the system
(\ref{eq:motiongral}),  invariant under the Killing vector field $\displaystyle \dot{w}_k= 2 i \, w_k $, is  the following system of algebraic equations,
\begin{equation} \label{eq:rationalsystem-4-sphere}
\frac{16 R^6(R^{2}-|w_{k,0}|^{2})w_{k,0}}{ (R^2+ |w_{k,0}|^2)^3  } =
 \sum_{j\neq k}^n \frac{m_j (w_{j,0}-w_{k,0})(R^2+\bar{w}_{j,0}w_{k,0})(R^2+|w_{j,0}|^2)}{|w_{j,0}-w_{k,0}|^3 
\, |R^2+ \bar{w}_{j,0} w_{k,0}|}.
\end{equation}
The velocity of  each particle is given by the relation $ \dot{w}_{k,0} = 2 i w_{k,0}$, for $k=1,2,\cdots, n$.
\end{Corollary}

\begin{proof} Let $w_k = w_k(t)= e^{2 it} \, w_{k,0}$ be the action of the Killing vector field $\displaystyle \dot{w}_k= 2 i \, w_k $ at  the initial condition point $w_{k,0}$, with velocity $ \dot{w}_{k,0} = 2 i w_{k,0}$. If we multiply equation (\ref{eq:rationalsystem-4-sphere})  by $ e^{2it}$, and use the equality $\bar{w}_j(t) w_k (t) = \bar{w}_{j,0} e^{-it} \, w_{k,0}e^{it} = \bar{w}_{j,0} \, w_{k,0}$, then is obtained the system,
\begin{equation} 
\frac{16 \, R^6(R^2-|w_k|^2) w_k }{ (R^2+ |w_k|^2)^4} =
 \sum_{j=1, j \neq k}^n \frac{m_j \, (|w_j|^2+R^2)^2(R^2+\bar{w}_j w_k)(w_j-w_k)}{|R^2+ \bar{w}_j w_k| \,|w_j-w_k|^3},
\end{equation}
which shows that $w_k(t)$ is a solution of (\ref{eq:rationalsystem-2}).

The converse claim follows directly  considering $t=0$ in system (\ref{eq:rationalsystem-2}). This proves the Corollary.
\end{proof}

In \cite{Perez} the reader can find examples for the two and three body problems defined on the conformal sphere $\mathbb{M}^2_{R}$.

%%%%%%%%%%%%%%%%%%%%%%%%%%%%%%%%%%%%%%
\section{Relative equilibria of the two-body problem  in $\mathbb{M}_R^2$}\label{sec:two-body}
%%%%%%%%%%%%%%%%%%%%%%%%%%%%%%%%%%%%%

In the case of $2$-bodies, invariant solutions will be shown under the field of vectors $\displaystyle \dot{w}_k= 2 i \, w_k $ such that the two masses move along two different circles.

Firstly, we note that the system (\ref{eq:rationalsystem-4-sphere}) for the two body problem becomes the simple algebraic system (independent of time $t$),
\begin{eqnarray} \label{eq:rationalsystem-two-body-esfere}
\frac{16 R^6(|w_{1}|^{2}-R^{2}) w_{1} }{ (R^{2}+ |w_{1}|^{2})^{3}} &=&
 \frac{m_{2} (w_{2}-w_{1})(R^{2}+\bar{w}_{2}w_{1})(R^{2}+|w_{2}|^{2})}{|w_{2}-w_{1}|^{3} |R^{2}+  w_{1} \bar{w}_{2}|}, \nonumber \\
 \frac{16 R^{6}(|w_{2}|^{2}-R^{2}) w_{2} }{ (R^{2}+ |w_{2}|^{2})^{3}} &=&
 \frac{m_{1} (w_{1}-w_{2})(R^{2}+\bar{w}_{1}w_{2}) (R^{2}+|w_{1}|^{2})}{|w_{1}-w_{2}|^{3} |R^{2}+ \bar{w}_{1} w_{2}|}. \nonumber \\
\end{eqnarray}

From Corollary \ref{coro:cyclical-2}, a necessary and sufficient condition for the existence of invariant elliptical solutions under the Killing vector field $\displaystyle \dot{w}_k= 2 i \, w_k $ is that the initial real positions  $w_1= \alpha$  and $w_2=\beta$ (with $0\leq \alpha, \,  \beta \leq R$) satisfy the system (\ref{eq:rationalsystem-two-body-esfere}). This is, substituting in such that system it becomes, 
\begin{eqnarray} \label{eq:rationalsystem-two-body-esf-real}
\frac{16 R^6(\alpha^{2}-R^{2}) \alpha }{ (R^{2}+ \alpha^{2})^{3}} &=&
 \frac{m_{2} (\beta -\alpha)(R^{2}+\beta \alpha)(R^{2}+\beta^{2})}{|\beta-\alpha|^{3} |R^{2}+  \alpha \beta|}, \nonumber \\
 \frac{16 R^{6}(\beta^{2}-R^{2}) \beta}{ (R^{2}+ \beta|^{2})^{3}} &=&
 \frac{m_{1} (\alpha-\beta)(R^{2}+\alpha \beta) (R^{2}+\alpha^{2})}{|\alpha-\beta|^{3} |R^{2}+ \alpha \beta|}. \nonumber \\
\end{eqnarray}

By dividing the right half hand and the left hand sides in system (\ref{eq:rationalsystem-two-body-esf-real}), and substituing $w_1=\alpha$, $w_2=\beta$ insuch system we obtain the relation
\begin{equation} \label{eq:rationalsystem-two-body-esf-real-2}
\frac{(\alpha^{2}-R^{2}) (R^{2}+ \beta^{2})^{2} \alpha }{(\beta^{2}-R^{2}) (R^{2}+ \alpha^{2})^{2} \beta} = - \frac{m_{2}} {m_{1}},
\end{equation}
or equivalently
\begin{equation}\label{polinomial-two-body}
m_{1}(\alpha^{2}-R^{2}) (R^{2}+ \beta^{2})^{2} \alpha + m_{2}(\beta^{2}-R^{2}) (R^{2}+ \alpha^{2})^{2} \beta =0.
\end{equation}

\begin{Theorem}(Borisov {\it et al.} \cite{Borisov2}, Ortega {\it et al} \cite{Ortega})\label{Th:relative-two-body} 
There are only two types of relative equilibria for the two body problem with equal masses, obtained as solutions of the system (\ref{eq:rationalsystem-two-body-esf-real}),
\begin{enumerate}
\item When both particles are sited in opposite sides of the same circle,  $\beta= - \alpha$, called {\it isosceles solutions}.
\item When the two particles are sited in different circles, $ \displaystyle \beta= \frac{R(\alpha-R)}{\alpha +R}$, forming right angle, called {\it right-angled solutions}. 
\end{enumerate}
Both solutions for $\beta$ are located in the real interval $(-R,0)$.
\end{Theorem}

\begin{proof}
The solutions for system (\ref{eq:rationalsystem-two-body-esf-real}) obtained from the equation (\ref{polinomial-two-body}) are given by
\[ \beta=-\alpha,  \qquad \beta= \frac{R(\alpha-R)}{\alpha +R} \]
as can be seen easily.
\smallskip

In the first case, $\beta= - \alpha$, we obtain the isosceles solutions.
\smallskip

For the second one, let $\ell$ be the geodesic distance from $\alpha$ to $ \displaystyle \beta= \frac{R(\alpha-R)}{\alpha +R}$. 

If we parametrize in coordinates $(x,y) \cong (w,\bar{w})$ of $\mathbb{M}^2_R$, the corresponding arc $\Gamma$ between them by,
\begin{eqnarray}
x(t) &=& t \nonumber \\
y(t) &=& 0 \nonumber \\
\end{eqnarray}
in the interval $\displaystyle \left[\frac{R(\alpha-R)}{\alpha +R}, \alpha \right]$, 
then

\begin{eqnarray}
\ell &=& \int_{\Gamma} d \ell= 2R^2 \int_{\frac{R(\alpha-R)}{\alpha +R}}^{\alpha}
\frac{dt}{R^2+t^2} \nonumber \\
&=& 2R \left[\arctan \left( \frac{\alpha}{R} \right) -
\arctan \left( \frac{R(\alpha-R)}{\alpha +R}\right) \right]  \nonumber \\
&=& 2R \arctan \left( \frac{\frac{\alpha}{R} - \frac{\alpha-R}{\alpha +R}}{1+\frac{\alpha}{R} \frac{\alpha-R}{\alpha +R}}\right) \nonumber \\
&=& 2R \arctan \, (1) = \frac{R \pi}{2} \nonumber \\
\end{eqnarray}

On the other hand we have that $\displaystyle \ell= R \theta$, where $\theta$ is the angle between the given points in $\mathbb{M}^2_R$. Therefore,  $\displaystyle \theta =\frac{\pi}{2}$, and the solution is of type right-angled.

This ends the proof.
\end{proof}

%%%%%%%%%%%%%%%%%%%%%%%%%%%%%%%%%%%%%%%%%%%%%%%%%%%%%%%%%%
\section{Conjugated equilibria solutions}\label{sec:conjugated points}
%%%%%%%%%%%%%%%%%%%%%%%%%%%%%%%%%%%%%%%%%%%%%%%%%%%%%%%%%%

From Theorem \ref{Th:relative-two-body} we have two types of relative equilibria and they generate two different types of equilibria when $\alpha \to R$.

\begin{Corollary}(Perez-Chavela et al. \cite{Perez})  For the right-angled relative equilibria we obtain in the limit when $\alpha \to R$ one equilibrium for the system when one of the particles is fixed at the origin of coordinates and the other is moving along the geodesic circle $|w|=R$ with velocity $\dot{w}(t)=-2i w(t)$.
\end{Corollary}

This become from the fact that 
\[ \frac{R(\alpha-R)}{\alpha +R} \to 0, \]
when $\alpha \to R$.

\smallskip

Now we begin the study of a solutions for conjugate (antipodal) points for the two body problem, as one limit case of isosceles solutions  when $\alpha \to R$. We have the following result which shows that the relative equilibria converge to the geodesic circle (equator) which behaves as one relative equilibrium because it is invariant under the canonical Killing vector field. It is not a solution of the general system (\ref{eq:motiongral}) even when it is a geodesic for $\mathbb{M}_R^2$ and the acting potential vanishes along such solution. Regardless, such that geodesic is a solution of the regularized system obtained from (\ref{eq:motiongral}) when we avoid the singularities in the force system due to conjugated antipodal points. 

\begin{Theorem}\label{main-Theorem}  For the isosceles relative equilibria we obtain in the limit when $\alpha \to R$ one equilibrium for the system when the particles with equal masses are geodesic conjugated points. Such that particles are moving along the geodesic circle $|w|=R$ with velocity $\dot{w}(t)=-2i w(t)$.
\end{Theorem}

\begin{proof}

We recall that for the two body problem with equal masses, the potential is 
\begin{equation}\label{eq:potesf-two-bod}
U_R (\mathbf{w},\mathbf{\bar{w}}) = \frac{m^2}{R^{2}} \left(  \frac{|R^2+ w_2 \bar{w}_1|}{|w_2-w_1|}\right), 
\end{equation}

Let $w_1(t)= \alpha e^{2it}$ and $w_2(t)= -\alpha e^{2it}$ be the components of the relative equilibrium 
\begin{equation}\label{eq;rel-equil}
w(t)= (w_1(t), w_2(t))
\end{equation}

Then equation (\ref{eq:potesf-two-bod}) becomes in
\begin{eqnarray}
U_R (\mathbf{w(t)},\mathbf{\bar{w}(t)}) &=& \frac{m^2}{R^{2}} \left(  \frac{|R^2+ (-\alpha e^{2it})(\alpha e^{-2it})|}{|2 \alpha e^{2it}|}  \right) \nonumber \\
&=& \frac{m^2}{R^{2}} \,  \frac{|R^2-\alpha^2|}{|2 \alpha|} \nonumber \\
&=& \frac{m^2 (R^2-\alpha^2)}{2 \alpha R^{2}} \nonumber \\
\end{eqnarray}

If we denote by $z_1(t)=  R e^{2it}$ and $w_2(t)= -R e^{2it}$ be the components of the function
\begin{equation}\label{conjugated-solution}
z(t)= (z_1(t), z_2(t)),
\end{equation} 
whose each of its coordinates parametrizes the geodesic circle of radius $R$ with velocity $\dot{z}_k(t)=-2i z_k(t)$,  then, 
\[ \lim_{\alpha \to R} w(t) = z(t). \]

From continuity of the potential, and since conjugated points are no longer singularities for the potential, we have
\begin{eqnarray}
 U_R (\mathbf{z(t)},\mathbf{\bar{z}(t)}) &=& \lim_{\alpha \to R} \,U_R (\mathbf{w_1(t)},\mathbf{\bar{w}_1(t)}) \nonumber \\
&=& \lim_{\alpha \to R} \frac{m^2 (R^2-\alpha^2)}{2 \alpha R^{2}} =0 \nonumber \\
\end{eqnarray}
which shows that the potential vanishes along (\ref{conjugated-solution}).

\smallskip

A simple substitution of $\alpha=-\beta=R$ in (\ref{polinomial-two-body}) shows that it also  a solution for that real equation, and for the regularized system, obtained from the relative equilibria condition system (\ref{eq:rationalsystem-two-body-esf-real}),
\begin{eqnarray} \label{eq:rationalsystem-two-body-esf-real-regularized}
\left(\frac{16 R^6(\alpha^{2}-R^{2}) \alpha }{ (R^{2}+ \alpha^{2})^{3}}\right)\, |R^{2}+  \alpha \beta| &=&  \frac{m_{2} (\beta -\alpha)(R^{2}+\beta \alpha)(R^{2}+\beta^{2})}{|\beta-\alpha|^{3} }, \nonumber \\
\left(\frac{16 R^{6}(\beta^{2}-R^{2}) \beta}{ (R^{2}+ \beta^{2})^{3}}\right)\,  |R^{2}+ \alpha \beta| &=&
 \frac{m_{1} (\alpha-\beta)(R^{2}+\alpha \beta) (R^{2}+\alpha^{2})}{|\alpha-\beta|^{3}}. \nonumber \\
\end{eqnarray}

It is easy to see that the function (\ref{conjugated-solution}) is not a solution of (\ref{eq:motiongral}) but another simple substitution shows that it satisfy the regularized system obtained of (\ref{eq:motiongral}) when we avoid singularities due also to conjugated antipodal points,
\begin{eqnarray}\label{eq:motiongral-two-body-regul}
\left( \ddot{w}_{1} - \frac{2\bar{w}_{1}\dot{w}_{1}^{2}}{R^{2}+|w_{1}|^{2}}\right)\, 
\frac{|R^{2}+\bar{w}_{2}w_{1}|}{(R^{2}+|w_{1}|^{2})^{2}} &=& \frac{m  \,(w_{2}-w_{1})(R^{2}+\bar{w}_{2}w_{1})(R^{2}+|w_{2}|^{2})}{4R^{6}|w_{2}-w_{1}|^{3}} \nonumber \\
\left( \ddot{w}_{2} - \frac{2\bar{w}_{2}\dot{w}_{2}^{2}}{R^{2}+|w_{2}|^{2}}\right)\,
\frac{ |R^{2}+\bar{w}_{1}w_{2}|}{(R^{2}+|w_{2}|^{2})^{2}} &=& \frac{m  \,(w_{1}-w_{2})(R^{2}+\bar{w}_{1}w_{2})(R^{2}+|w_{1}|^{2})}{4R^{6}|w_{1}-w_{2}|^{3}} \nonumber \\
\end{eqnarray}

This ends the proof.
\end{proof}

\begin{Definition}
We call to the solution (\ref{conjugated-solution}) of the regularized system (\ref{eq:motiongral-two-body-regul}), invariant under the Killing vector field $\dot{z}= 2i z$ a conjugated equilibrium solution for the two body problem in $\mathbb{M}_R^2$.
\end{Definition}

We remark the the possibility of having a relative equilibrium in antipodal conjugated points was claimed in (Perez-Chavela {\it et al.} \cite{Perez}), but in that case these points were singularities for the potential there used.

%%%%%%%%%%%%%%%%%%%%%%%%%%%%%%%%%%%%%%%%%%%%%%%%%%%%%%%%%%
\subsection*{Funding Sources}
%%%%%%%%%%%%%%%%%%%%%%%%%%%%%%%%%%%%%%%%%%%%%%%%%%%%%%%%%%
%%%%%%%%%%%%%%%%%%%%%%%%%%%%%%%%%%%%%%%%%%%%%%%%%%%%%%%%%%%%%%%%%%

This work was supported by the Universidad de Cartagena, Cartagena, Colombia [Grant Number 066-2013] and The Universidad Aut\'onoma Metropolitana Iztapalapa, Cd de M\'exico, M\'exico.

\end{document}